\documentclass[10pt,reqno]{amsart}
\usepackage{bbm}
\usepackage{tabu}
\usepackage{amsmath}
\usepackage{mathrsfs}
\usepackage{bm}
\usepackage{amsfonts} 
\usepackage{graphicx,latexsym,bm,amsmath,amssymb,verbatim,multicol,lscape}
\usepackage{enumerate}
\usepackage{xcolor}

\usepackage{todonotes}

\newtheorem{thm}{Theorem} [section]

\newtheorem{cor}[thm]{Corollary}
\newtheorem{lem}[thm]{Lemma}
\newtheorem{prop}[thm]{Proposition}

\theoremstyle{definition}
\newtheorem{defn}[thm]{Definition}
\theoremstyle{remark}
\newtheorem{rem}[thm]{Remark}
\newtheorem{exa}[thm]{Example}

\numberwithin{equation}{section}

\newcommand{\fq}{{\mathbb F}_{q}}
\newcommand{\fp}{{\mathbb F}_{p}}

\newcommand{\C}{{\mathcal C}}

\newcommand{\Tr}{\mbox{Tr}}
\newcommand{\wt}{\mbox{wt}}

\newcommand{\rmv}[1]{}

\def\<{\left\langle}
\def\>{\right\rangle}

\begin{document}

\title[Weight distribution of a class of $p$-ary codes]{Weight distribution of a class of $p$-ary codes}%
\author{Kaimin Cheng}
\address{School of Mathematics and Information, China West Normal University,
Nanchong, 637002, P. R. China}
\email{ckm20@126.com,kcheng1020@gmail.com}
\author{Du Sheng}
\address{School of Mathematics and Information, China West Normal University,
Nanchong, 637002, P. R. China}
\email{dsheng@cwnu.edu.cn}
\thanks{This work was supported partially the China Scholarship Council Foundation(No. 202301010002) and the
Scientific Research Innovation Team Project of China West Normal University (No. KCXTD2024-7).}
\date{\today}
\begin{abstract}
Let $p$ be a prime, and let $N$ be a positive integer such that $p$ is a primitive root modulo $N$. Define $q = p^e$, where $e = \phi(N)$, and let $\mathbb{F}_q$ be the finite field of order $q$ with $\mathbb{F}_p$ as its prime subfield. Denote by $\mathrm{Tr}$ the trace function from $\mathbb{F}_q$ to $\mathbb{F}_p$. For $\alpha \in \mathbb{F}_p$ and $\beta \in \mathbb{F}_q$, let $D$ be the set of nonzero solutions in $\mathbb{F}_q$ to the equation $\mathrm{Tr}(x^{\frac{q-1}{N}} + \beta x) = \alpha$. Writing $D = \{d_1, \ldots, d_n\}$, we define the code $\mathcal{C}_{\alpha,\beta} = \{(\mathrm{Tr}(d_1 x), \ldots, \mathrm{Tr}(d_n x)) : x \in \mathbb{F}_q\}$. In this paper, we investigate the weight distribution of $\mathcal{C}_{\alpha,\beta}$ for all $\alpha \in \mathbb{F}_p$ and $\beta \in \mathbb{F}_q$, with a focus on general odd primes $p$. When $\beta = 0$, we establish that $\mathcal{C}_{\alpha,0}$ is a two-weight code for any $\alpha \in \mathbb{F}_p$ and compute its weight distribution. For $\beta \neq 0$, we determine all possible weights of codewords in $\mathcal{C}_{\alpha,\beta}$, demonstrating that it has at most $p+1$ distinct nonzero weights. Additionally, we prove that the dual code $\mathcal{C}_{0,0}^{\perp}$ is optimal with respect to the sphere packing bound. These findings extend prior results to the broader case of any odd prime $p$.
\end{abstract}

\maketitle

\section{Introduction}
Let $\fq$ the finite field with $q$ elements and denote by $\fq^*$ the multiplicative group of $\fq$, where $q$ is a power of a prime $p$. Let $n,k$, and $d$ be positive integers. An $[n,k,d]$ $p$-ary linear code is a $k$-dimensional subspace of $\mathbb{F}_p^n$ with minimum (Hamming) distance $d$. For a positive integer $i$, denote by $A_i$ the number of codewords with Hamming weight $i$ in a code $\mathcal{C}$ of length $n$. The weight enumerator of $\mathcal{C}$ is given by
$$1+A_1z+A_2z^2+\cdots+A_nz^n,$$
and the sequence $(A_1,A_2,\ldots,A_n)$ is called the \textit{weight distribution} of $\mathcal{C}$. If the number of nonzero $A_i$ in the sequence $(A_1,A_2,\ldots,A_n)$ is $v$, then the code $\mathcal{C}$ is called a $v$-weight code. The weight distribution of a linear code determines both the minimum distance and error-correcting capabilities of the code. Moreover, it provides vital information for computing the probability of error detection and correction for certain algorithms \cite{[Klo]}. Codes with few weights have garnered significant interest due to their applications in various fields, including secret sharing \cite{[ADHK],[Sha],[YD]}, authentication codes \cite{[DW]}, and combinatorial designs \cite{[CK]}. Therefore,  constructing such codes and analyzing the weight distribution of these codes is an interesting and meaningful topic in coding theory.

Given a set $\{d_1, \ldots, d_n\} \subseteq \mathbb{F}_q^*$ (termed a defining set), we define a code $\mathcal{C}$ as follows:
\begin{align}\label{c1-1}
\mathcal{C} := \{({\rm Tr}(d_1x), {\rm Tr}(d_2x), \ldots, {\rm Tr}(d_nx)) : x \in \mathbb{F}_q\},
\end{align}
where ${\rm Tr}$ represents the trace function from $\mathbb{F}_q$ to $\mathbb{F}_p$. Clearly, $\mathcal{C}$ is a $p$-ary linear code. This approach to constructing linear codes via a defining set is fundamental, as every linear code over $\mathbb{F}_p$ can be expressed using some defining set (possibly a multiset) \cite{[Din]}. For examples of code classes constructed through careful selection of defining sets, see \cite{[Din1], [Din2], [TXF], [WDX]}.

Let $N$ be a positive integer satisfying $p$ being a primitive root modulo $N$ (so that $N$ must be one of the forms $2$, $4$, $\ell^m$ and $2\ell$, where $\ell$ is an odd prime and $m$ is a positive integer), and assume that $q=p^{\phi(N)}$, where $\phi$ is the Euler totient function.
For $\alpha\in\fp$ and $\beta\in\fq$, define a defining set $D$ by
\begin{align}\label{c1-2}
D=\left\{x\in\fq^*:\ {\rm Tr}(x^{\frac{q-1}{N}}+\beta x)=\alpha\right\},
\end{align}
and let $\C_{\alpha,\beta}$ be constructed as \eqref{c1-1} when $D=\{d_1,\ldots,d_n\}$. In 2015, by using some results of Moisio \cite{[Moi]}, Wang et al. \cite{[WDX]} proved that $\C_{0,0}$ is a two-weight code in the case of $p=2$ and $N=\ell^m$. In 2024, the first author, together with Gao \cite{[CG]}, proposed a novel method for evaluating a class of binomial Weil sums. Building on the precise computation of these Weil sums results, they determined the weight distribution of $\C_{0,0}$ for $p = 3$ and $N=\ell^m$. More recently, for $p = 3$, the first author \cite{[Che]}, leveraging the findings from \cite{[CG]}, established the weight distribution of $\C_{\alpha,\beta}$ for any $\alpha \in \mathbb{F}_p$ and $\beta \in \mathbb{F}_q$. In this paper, we seek to extend this analysis by determining the weight distribution of $\C_{\alpha,\beta}$ for any $\alpha \in \mathbb{F}_p$ and $\beta \in \mathbb{F}_q$ for the general odd prime $p$. Actually, when $\beta=0$, we prove that $\C_{\alpha,0}$ is a two-weight code for any $\alpha\in\fp$ and provide its weight distribution; when $\beta\ne 0$, we present all the possible weights of codewords in $\C_{\alpha,\beta}$, showing that $\C_{\alpha,\beta}$ has at most $p+1$ distinct nonzero weights. Clearly, our results extend previous findings for $p=2,3$ to the general case of any odd prime $p$.

In the following of this paper, we consistently assume that $N = 2\ell^m$ such that $p$ is primitive modulo $2\ell^{m}$, where $\ell$ is a prime distinct with $p$ and $m$ is a positive integer. Since the cases of other forms ($N=4, \ell^{m}$) of $N$ can be addressed similarly, we focus exclusively on the most complex case. The paper is structured as follows. In Section 2, we present essential preliminaries. In Section 3, by deriving precise evaluations of two Weil sums, we determine the weight distribution of a class of $p$-ary codes. In Section 4, we analyze the dual codes and obtain a class of optimal codes.

\section{Preliminaries}
In this section, we provide the necessary preliminaries and background results. Let $p$ be an odd prime, and $q=p^e$ with $e=\phi(2\ell^m)$. Denote by $\chi$ the canonical additive character of $\fq$, i.e., $\chi(c)=\zeta_p^{{\rm Tr}(c)}$ for any $c\in\fq$, where $\zeta_p=\exp(2\pi i/p)\in\mathbb{C}$. For $a,b\in\fq$, define two exponential sums over $\fq$ by
\begin{align}\label{c2-1}
S_{2\ell^{m}}(a,b):=\sum_{x\in\fq^*}\chi(ax^{\frac{q-1}{2\ell^{m}}}+bx)\ \text{and}\ S_{2\ell^{m}}(a):=\sum_{i=0}^{2\ell^{m}-1}\chi(a\xi^{i}).
\end{align}
Let $g$ be a fixed primitive element of $\mathbb{F}_q$. Define $\xi := g^{\frac{q-1}{N}}$, and note that one can verify $\mathbb{F}_q = \mathbb{F}_p(\xi)$, implying that ${\xi, \ldots, \xi^e}$ forms a basis of $\mathbb{F}_q$ over $\mathbb{F}_p$. The following fundamental result from \cite{[CG]} offers a method for computing $S_{2\ell^{m}}(a,b)$, which enables explicit evaluations of $S_{2\ell^{m}}(a,b)$ (see Proposition \ref{prop3.1}).
\begin{lem}\label{lem2.1}
The following statements are true.
\begin{enumerate}[(a)]
\item  For any $a\in\fq$, we have 
$$S_{2\ell^{m}}(a,0)=\frac{q-1}{2\ell^{m}}S_{2\ell^{m}}(a).$$
\item  For any $a\in\fq$ and $b\in\fq^*$, we have 
$$S_{2\ell^{m}}(a,b)=\chi(ab^{-\frac{q-1}{2\ell^{m}}})\sqrt{q}-
\frac{\sqrt{q}+1}{2\ell^{m}}
S_{2\ell^{m}}(ab^{-\frac{q-1}{2\ell^{m}}}).$$
\end{enumerate}
\end{lem}

Define $J$ to be the set of nonnegative integers $j$ with $ 0\le j\le 2\ell^{m}-1$. Let $J_{1},J_{2},J_{3}$ and $J_{4}$ be subsets of $J$ defined by
\begin{align*}
J_1&=\{j\in J:\ 1\le j\le (\ell-1)\ell^{m-1}\},\\
J_2&=\{j\in J:\  (\ell-1)\ell^{m-1}< j\le\ell^{m}\},\\
J_3&=\{j\in J:\  \ell^{m}< j\le 2\ell^{m}-\ell^{m-1}\},\\
J_4&=\{j\in J:\  2\ell^{m}-\ell^{m-1}< j\le\ 2\ell^{m}-1\}.
\end{align*}
Clearly, $(\{0\},J_1,J_2,J_3,J_4)$ is a partition of $J$. Now we partition $J_1$ into $J_{1}^{(1)}\cup J_{1}^{(2)}\cup J_{1}^{(3)}$, where $J_1^{(1)}=\{j\in J_1: \ell^{m-1}\nmid j\}$, $J_1^{(2)}=\{j\in J_1: \ell^{m-1}\mid j\ \text{and}\ 2\nmid j\}$ and $J_1^{(3)}=\{j\in J_1: 2\ell^{m-1}\mid j\}$.  Note that $J_3=\{\ell^m+k\ell^{m-1}-u:\ 1\le k\le \ell-1\ \text{and}\ 0\le u\le\ell^{m-1}-1\}$.
So we can partition $J_3$ into $J_3=J_3^{(1)}\cup J_3^{(2)}\cup J_3^{(3)}$, where
$J_3^{(1)}=\{\ell^m+k\ell^{m-1}-u:\ 1\le k\le \ell-1\ \text{and}\ 0<u\le\ell^{m-1}-1\}$, $J_3^{(2)}=\{\ell^m+k\ell^{m-1}:\ 1\le k\le \ell-1\ \text{and}\ 2\mid k\}$,  and $J_3^{(3)}=\{\ell^m+k\ell^{m-1}:\ 1\le k\le \ell-1\ \text{and}\ 2\nmid k\}$.

Recall that $\xi = g^{\frac{q-1}{2\ell^m}}$ is a primitive $2\ell^m$-th root of unity in $\mathbb{F}_q$.  The following result, presented in \cite[Lemma 2.8]{[CG]}, gives the trace of any power of $\xi$.
\begin{lem}\label{lem2.2}
Let $j\in J$. Then
$${\rm Tr}(\xi^j)=\begin{cases}
(\ell-1)\ell^{m-1},& \text{if}\ j=0,\\
-(\ell-1)\ell^{m-1},& \text{if}\ j=\ell^{m},\\
-\ell^{m-1},& \text{if}\ j\in J_{1}^{(3)}\cup J_{3}^{(3)},\\
\ell^{m-1},& \text{if}\ j\in J_{1}^{(2)}\cup J_{3}^{(2)},\\
0,& \text{otherwise}.
\end{cases}$$
\end{lem}

\begin{cor}\label{cor2.3}
Let $S_{2\ell^{m}}(a)$ be the second sum defined as in \eqref{c2-1}. The following statements are true.
\begin{enumerate}[(a)]
\item \cite[Corollary 2.3]{[Che]} For any $a\in\fp$, we have
$$S_{2\ell^m}(a)=\zeta_p^{\ell^{m-1}(\ell-1)a}+
\zeta_p^{-\ell^{m-1}(\ell-1)a}+(\ell-1)(\zeta_p^{\ell^{m-1}a}+
\zeta_p^{-\ell^{m-1}a})+2\ell^m-2\ell.$$
\item For any $j\in J$ and any $a\in\fp$, we have $S_{2\ell^m}(a\xi^{j})=S_{2\ell^m}(a)$.
\end{enumerate}
\end{cor}
\begin{proof}
Item (b) is immediate from the definition of $S_{2\ell^{m}}(a)$.
\end{proof}

\section{The weight distribution of $\C_{\alpha,\beta}$}

Let $\alpha\in\fp$ and $\beta\in\fq$. Let $D$ be defined as in \eqref{c1-2}, and let the code $\C_{{\alpha,\beta}}$ be defined as in \eqref{c1-1}. In this ection, we shall present results on the weight distribution of $\C_{\alpha,\beta}$.  Recall that $g$ is the primitive element of $\fq$ and $\xi=g^{\frac{q-1}{2\ell^{m}}}$. For $b\in\fq^{*}$, let ${\rm Ind}(b)$ be the index of $b$ with respect to the base $g$, and define $j_{b}$ to be the least nonnegative residue of $-{\rm Ind}(b)$ modulo $2\ell^{m}$. First, we provide the following result on evaluations of $S_{2\ell^{m}}(a,b)$.
\begin{prop}\label{prop3.1}
Let $b\in\fq^{*}$. Each of the following is true.
\begin{enumerate}[(a)]
\item We have 
$$S_{2\ell^{m}}(1,b)=\begin{cases}
\zeta_{p}^{(\ell-1)\ell^{m-1}}\sqrt{q}-\frac{\sqrt{q}+1}{2\ell^{m}}S_{2\ell^{m}}(1),&\text{if}\ j_{b}=0,\\
\zeta_{p}^{-(\ell-1)\ell^{m-1}}\sqrt{q}-\frac{\sqrt{q}+1}{2\ell^{m}}S_{2\ell^{m}}(1),& \text{if}\ j_{b}=\ell^{m},\\
\zeta_{p}^{-\ell^{m-1}}\sqrt{q}-\frac{\sqrt{q}+1}{2\ell^{m}}S_{2\ell^{m}}(1),& \text{if}\ j_{b}\in J_{1}^{(3)}\cup J_{3}^{(3)},\\
\zeta_{p}^{\ell^{m-1}}\sqrt{q}-\frac{\sqrt{q}+1}{2\ell^{m}}S_{2\ell^{m}}(1),& \text{if}\ j_{b}\in J_{1}^{(2)}\cup J_{3}^{(2)},\\
\sqrt{q}-\frac{\sqrt{q}+1}{2\ell^{m}}S_{2\ell^{m}}(1),& \text{otherwise},
\end{cases}$$
where $S_{2\ell^{m}}(1)=\zeta_p^{\ell^{m-1}(\ell-1)}+
\zeta_p^{-\ell^{m-1}(\ell-1)}+(\ell-1)(\zeta_p^{\ell^{m-1}}+
\zeta_p^{-\ell^{m-1}})+2\ell^m-2\ell$.
\item If we treat $S_{2\ell^{m}}(1,b)$ as a function of $\zeta_{p}$, denoted by $S_{2\ell^{m}}(1,b;\zeta_{p})$, then 
$$S_{2\ell^{m}}(a,b)=S_{2\ell^{m}}(1,b;\zeta_{p}^{a})$$
for any $a\in\fp^{*}$.
\item For any $a\in\fq$ and any $x\in\fp^{*}$ we have $S_{2\ell^{m}}(a,b)=S_{2\ell^{m}}(a,xb)$.
\end{enumerate}
\end{prop}
\begin{proof}
Let $b \in\fq^{*}$, and let $j_b$ be the least nonnegative residue of $-\mathrm{Ind}(b)$ modulo $2\ell^m$. Then, we have $b^{-\frac{q-1}{2\ell^m}} = (g^{\frac{q-1}{2\ell^m}})^{-\mathrm{Ind}(b)} = \xi^{-\mathrm{Ind}(b)} = \xi^{j_b}$, since $\xi$ is a primitive $2\ell^m$-th root of unity. Thus, Item (a) follows directly from Lemma \ref{lem2.1}(a) and Lemma \ref{lem2.2}. For Item (b), consider $a \in \mathbb{F}_p^{*}$. It follows that $a b^{-\frac{q-1}{2\ell^m}} = a \xi^{j_b}$. On the one hand, we compute $\chi(a \xi^{j_b}) = \zeta_p^{\mathrm{Tr}(a \xi^{j_b})} = \zeta_p^{a \mathrm{Tr}(\xi^{j_b})}$. On the other hand, Corollary \ref{cor2.3}(b) implies that $S_{2\ell^m}(a \xi^{j_b}) = S_{2\ell^m}(a)$. Therefore, Item (b) is established by Lemma \ref{lem2.1}(b) and Corollary \ref{cor2.3}(a). Finally, for Item (c), it suffices to show that $x^{-\frac{q-1}{2\ell^m}} = 1$ for all $x \in \mathbb{F}_p^{*}$, which is equivalent to the condition $(p-1) \mid \frac{q-1}{2\ell^m}$. Since $q = p^e$ with $e$ being even, we can express $\frac{q-1}{2\ell^m} = \frac{(\sqrt{q} + 1)(\sqrt{q} - 1)}{2\ell^m}$ and note that $\gcd(\sqrt{q} + 1, \sqrt{q} - 1) = 2$. This implies that $\ell^m$ must divide either $\sqrt{q} + 1$ or $\sqrt{q} - 1$. Suppose $\ell^m \mid (\sqrt{q} - 1)$; then $2\ell^m \mid (\sqrt{q} - 1)$. This would imply that the order of $p$ modulo $2\ell^m$ is at most $\frac{e}{2}$, contradicting the assumption that $p$ is a primitive root modulo $2\ell^m$. Hence, $\ell^m \mid (\sqrt{q} + 1)$, so $\frac{\sqrt{q} + 1}{2\ell^m}$ is an integer. Consequently, $(p-1) \mid \frac{(\sqrt{q} + 1)(\sqrt{q} - 1)}{2\ell^m}$, since $(p-1) \mid (\sqrt{q} - 1)$. This completes the proof of Proposition \ref{prop3.1}.
\end{proof}

For  $x\in\fq$, define a useful notation $w(\alpha,x)$ by
\begin{align}\label{c3-1}
w(\alpha,x)=\sum_{y\in\fp^{*}}\zeta_{p}^{-y\alpha}S_{2\ell^{m}}(y,xy).
\end{align}
Clearly, from Proposition \ref{prop3.1}(c) we know that $w(\alpha,xz)=w(\alpha,x)$ for any $z\in\fp^{*}$ and $x\in\fq$. Moreover, we have the following explicit results.
\begin{prop}\label{prop3.2}
Let $\alpha\in\fp$ and $x\in\fq$. Then the explicit value of $w(\alpha,x)$ is given as follows.
\begin{enumerate}[(a)]
\item When $\alpha=x=0$. Then we have
$$w(0,0)=\begin{cases}
\frac{q-1}{\ell^{m}}\left((p-1)\ell^{m}-p\ell+p\right),&\text{if}\ \ell\equiv 1\pmod{p},\\
\frac{q-1}{\ell^{m-1}}\left((p-1)\ell^{m-1}-p\right),&\text{if}\ \ell\not\equiv 1\pmod{p}.
\end{cases}$$
\item When $\alpha=0$ and $x\ne 0$. The two statements are true.
\begin{enumerate}[(i)]
\item If $\ell\equiv 1\pmod{p}$, then
$$w(0,x)=\begin{cases}
\frac{\sqrt{q}+1}{\ell^{m}}\left(-p\ell^{m}+p\ell-p\right)+1,&\text{if } j_{x} \in J_{1}^{(2)} \cup J_{3}^{(2)} \cup J_{1}^{(3)} \cup J_{3}^{(3)},\\
\frac{\sqrt{q}+1}{\ell^{m}}\left(p\ell-p\right)-p+1,&\text{otherwise}.
\end{cases}$$
\item If $\ell\not\equiv 1\pmod{p}$, then
$$w(0,x)=\begin{cases}
\frac{p(\sqrt{q}+1)}{\ell^{m-1}}-p+1,&\text{if } j_{x} \in J_{1}^{(1)} \cup J_{3}^{(1)} \cup J_{2}\setminus\{\ell^{m}\},\\
\frac{\sqrt{q}+1}{\ell^{m-1}}\left(-p\ell^{m-1}+p\right)+1,&\text{otherwise}.
\end{cases}$$
\end{enumerate}
\item When $\alpha\ne 0$. Let $a_{1}=\ell^{m-1}(\ell-1)-\alpha$, $a_{2}=-\ell^{m-1}(\ell-1)-\alpha$, $a_{3}=\ell^{m-1}-\alpha$ and $a_{4}=-\ell^{m-1}-\alpha$. The following  statements are true.
\begin{enumerate}[(i)]
\item If $x=0$, then
$$w(\alpha,0)=\begin{cases}
1-q,&\text{if}\ a_{1}, a_{2}, a_{3}\ \text{and}\ a_{4}\ \text{are\ nonzero} ,\\
\frac{q-1}{2\ell^{m}}\left(-2\ell^{m}+p\right),&\text{if}\ \text{one of}\ a_{1}\ \text{and}\ a_{2}\ \text{is\ zero}\ \text{but}\ a_{3}\ne 0, a_{4}\ne 0,\\
\frac{q-1}{2\ell^{m}}\left(-2\ell^{m}+p\ell-p\right),&\text{if}\ \text{one of}\ a_{3}\ \text{and}\ a_{4}\ \text{is\ zero}\ \text{but}\ a_{1}\ne 0, a_{2}\ne 0,\\
\frac{q-1}{2\ell^{m-1}}\left(-2\ell^{m-1}+p\right),&\text{if}\ a_{1}=a_{3}=0\ \text{or}\ \ a_{2}=a_{4}=0.\end{cases}$$
\item Let $x\ne 0$. Then $w(\alpha,x)=1$ if none of $\ a_{1}, a_{2}, a_{3} $ and $a_{4}$ are zero. Furthermore, we have the following results.
\begin{enumerate}[(1)]
\item If $a_{1}=0$ but all of $a_{2},a_{3}$ and $a_{4}$ are nonzero, then
$$w(\alpha,x)=\begin{cases}
\frac{p(\sqrt{q}+1)}{2\ell^{m}}(2\ell^{m}-1)-p+1,&\text{if } j_{x}=0,\\
\frac{-p(\sqrt{q}+1)}{2\ell^{m}}+1,&\text{otherwise}.
\end{cases}$$
\item If $a_{2}=0$ but all of $a_{1},a_{3}$ and $a_{4}$ are nonzero, then
$$w(\alpha,x)=\begin{cases}
\frac{p(\sqrt{q}+1)}{2\ell^{m}}(2\ell^{m}-1)-p+1,&\text{if } j_{x}=\ell^{m},\\
\frac{-p(\sqrt{q}+1)}{2\ell^{m}}+1,&\text{otherwise}.
\end{cases}$$
\item If $a_{3}=0$ but all of $a_{1},a_{2}$ and $a_{4}$ are nonzero, then
$$w(\alpha,x)=\begin{cases}
\frac{p(\sqrt{q}+1)}{2\ell^{m}}(2\ell^{m}-\ell+1)-p+1,&\text{if } j_{x}\in J_{1}^{(2)}\cup J_{3}^{(2)},\\
\frac{-p(\sqrt{q}+1)}{2\ell^{m}}(\ell-1)+1,&\text{otherwise}.
\end{cases}$$
\item If $a_{4}=0$ but all of $a_{1},a_{2}$ and $a_{3}$ are nonzero, then
$$w(\alpha,x)=\begin{cases}
\frac{p(\sqrt{q}+1)}{2\ell^{m}}(2\ell^{m}-\ell+1)-p+1,&\text{if } j_{x}\in J_{1}^{(3)}\cup J_{3}^{(3)},\\
\frac{-p(\sqrt{q}+1)}{2\ell^{m}}(\ell-1)+1,&\text{otherwise}.
\end{cases}$$
\item If $a_{1}=a_{3}=0$ but $a_{2}\ne 0, a_{4}\ne 0$, then
$$w(\alpha,x)=\begin{cases}
\frac{p(\sqrt{q}+1)}{2\ell^{m-1}}(2\ell^{m-1}-1)-p+1,&\text{if } j_{x}\in \{0\}\cup J_{1}^{(2)}\cup J_{3}^{(2)},\\
\frac{-p(\sqrt{q}+1)}{2\ell^{m-1}}+1,&\text{otherwise}.
\end{cases}$$
\item If $a_{2}=a_{4}=0$ but $a_{1}\ne 0, a_{3}\ne 0$, then
$$w(\alpha,x)=\begin{cases}
\frac{p(\sqrt{q}+1)}{2\ell^{m-1}}(2\ell^{m-1}-1)-p+1,&\text{if } j_{x}\in \{\ell^{m}\}\cup J_{1}^{(3)}\cup J_{3}^{(3)},\\
\frac{-p(\sqrt{q}+1)}{2\ell^{m-1}}+1,&\text{otherwise}.
\end{cases}$$
\end{enumerate}
\end{enumerate}
\end{enumerate}
\end{prop}
\begin{proof}
First, let us prove Item (a). By the definition and Lemma \ref{lem2.1}(a) we have
$$w(0,0)=\sum_{x\in\fp^*}S_{2\ell^m}(y,0)=\frac{q-1}{2\ell^m}\sum_{y\in\fp^*}S_{2\ell^m}(y).$$
Applying Corollary \ref{cor2.3}(a) and noting that 
\begin{align}\label{c3-2}
\sum_{y\in\fp^*}\zeta_{p}^{yz}=\begin{cases}
p-1,&\text{if}\ z=0,\\
-1,&\text{if}\ z\in\fp^*,
\end{cases}
\end{align}
the desired result follows.

For Item (b), let $\alpha=0$ and $x\in\fq^*$. By the definition and Proposition \ref{prop3.1}(b) we have
$$w(0,x)=\sum_{y\in\fp^*}S_{2\ell^m}(y,x)=\sum_{y\in\fp^*}S_{2\ell^m}(1,x;\zeta_p^y).$$
It then follows from Lemma \ref{lem2.1}(b) and Corollary \ref{cor2.3}(b) that
$$w(0,x)=\sum_{y\in\fp^*}\left(\sqrt{q}\zeta_p^{y{\rm Tr}(\xi^{j_x})}-\frac{\sqrt{q}+1}{2\ell^m}S_{2\ell^m}(y)\right).$$
Thus, by substituting the results of Lemma \ref{lem2.2} and Corollary \ref{cor2.3}(a), along with \eqref{c3-2}, the result is obtained.

Finally, we prove Item (c). Let $\alpha\ne 0$, and let $a_{1}=\ell^{m-1}(\ell-1)-\alpha$, $a_{2}=-\ell^{m-1}(\ell-1)-\alpha$, $a_{3}=\ell^{m-1}-\alpha$, $a_{4}=-\ell^{m-1}-\alpha$. For such elements $a_1, a_2, a_3, a_4 \in \fp$, one observes that there are only three possible cases: (1) all are nonzero; (2) exactly one is zero while the others are nonzero; (3) either $a_1 = a_3 = 0$ with $a_2 \neq 0$ and $a_4 \neq 0$, or $a_2 = a_4 = 0$ with $a_1 \neq 0$ and $a_3 \ne 0$. First, by the definition and lemma \ref{lem2.1}(a) we have
\begin{align}\label{c3-3}
w(\alpha,0)=\sum_{y\in\fp^*}\zeta_p^{-y\alpha}S_{2\ell^m}(y,0)=\frac{q-1}{2\ell^m}\sum_{y\in\fp^*}\zeta_p^{-y\alpha}S_{2\ell^m}(y).
\end{align}
Putting the evaluation of $S_{2\ell^m}(y)$ presented in Lemma \ref{cor2.3}(a) into \eqref{c3-3}, we derive that
\begin{align*}w(\alpha,0)=&\frac{q-1}{2\ell^m}\sum_{y\in\fp^*}\Big(\zeta_p^{y(\ell^{m-1}(\ell-1)-\alpha)}+
\zeta_p^{y(-\ell^{m-1}(\ell-1)-\alpha)}\\
&+(\ell-1)(\zeta_p^{y(\ell^{m-1}-\alpha)}+\zeta_p^{y(-\ell^{m-1}-\alpha)})+(2\ell^m-2\ell)\zeta_p^{-y\alpha}\Big).
\end{align*}
By using Equation \eqref{c3-2}, the result of $w(\alpha,0)$ is obtained. Now let $x\in\fq^*$ and let $j_x$ be defined as before. From the definition as in \eqref{c3-1}, Lemma \ref{lem2.1}(b) and Corollary \ref{cor2.3}, we have
\begin{align*}
w(\alpha,x)&=\sqrt{q}\sum_{y\in\fp^*}\zeta_p^{-y\alpha}\chi(y\xi^{j_x})-\frac{\sqrt{q}+1}{2\ell^m}\sum_{y\in\fp^*}\zeta_p^{-y\alpha}S_{2\ell^m}(y\xi^{j_x})\\
&=\sqrt{q}\sum_{y\in\fp^*}\left(\zeta_p^{{\rm Tr}(\xi^{j_x})-\alpha}\right)^{y}-\frac{\sqrt{q}+1}{q-1}w(\alpha,0).
\end{align*}
The desired results follow from Lemma \ref{lem2.2}, the result of $w(\alpha,0)$ and Equation \eqref{c3-2}. 

It completes the proof of Proposition \ref{prop3.2}.
\end{proof}

\begin{prop}\label{prop3.3}
Let $\C_{{\alpha,\beta}}$ be the code whose defining set $D$ is given in \eqref{c1-2}. The following statements hold.
\begin{enumerate}[(a)]
\item The code $\C_{{\alpha,\beta}}$ is of length
$$n=\frac{1}{p}\left(q-1+w(\alpha,\beta)\right).$$
\item Let $D=\{d_{1},\ldots,d_{n}\}$ and $c_{\gamma}=({\rm Tr}(\gamma d_{1}),\ldots,{\rm Tr}(\gamma d_{n}))\in\C_{\alpha,\beta}$ with $\gamma\in\fq^{*}$. Then the Hamming weight of $c_{\gamma}$, denoted by $\text{wt}(c_{\gamma})$, is
$$\text{wt}(c_{\gamma})=\frac{1}{p^{2}}\Big((p-1)q+(p-1)w(\alpha,\beta)-\sum_{z\in\fp^{*}}w(\alpha,\beta+z\gamma)\Big).$$
\end{enumerate}
\end{prop}
\begin{proof}
We first determine $n$, the length of the code $\mathcal{C}_{\alpha,\beta}$, defined as $n = \#D$. To do this, we introduce $n_0=\#\{x\in\fq: \Tr(x^{\frac{q-1}{2\ell^m}}+\beta x)=\alpha\}$ and compute that
\begin{align*}
n_{0}&=\frac{1}{p}\sum_{x\in\fq}\sum_{y\in\fp}\zeta_{p}^{y({\rm Tr}(x^{\frac{q-1}{2\ell^m}}+\beta x)-\alpha)}\\
&=\frac{q}{p}+\frac{1}{p}\sum_{y\in\fp^*}\zeta_{p}^{-y\alpha}
\Big(\sum_{x\in\fq^*}\chi(yx^{\frac{q-1}{2\ell^m}}+y\beta x)+1\Big)\\
&=\frac{q}{p}+\frac{1}{p}\sum_{y\in\fp^*}\zeta_{p}^{-y\alpha}
\Big(S_{2\ell^m}(y,y\beta)+1\Big).
\end{align*}
By using Equation \eqref{c3-2} and noting that
$$n=\begin{cases}
n_{0}-1,&\text{if}\ \alpha=0,\\
n_{0},&\text{if}\ \alpha\ne 0.
\end{cases}$$
We then have $n=\frac{1}{p}\left(q-1+w(\alpha,\beta)\right)$, as desired. 

Next, let $c_{\gamma}:=(\Tr(\gamma d_1),\ldots,\Tr(\gamma d_n))\in\mathcal{C}_{\alpha,\beta}$ with $\gamma\in\fq^{*}$ and $\text{wt}(c_{\gamma})$ be the Hamming weight of $c_{\gamma}$. We introduce a notation $N_{\gamma}$ defined by
$$N_{\gamma}:=\#\{x\in\fq:\ \Tr(x^{\frac{q-1}{2\ell^m}}+\beta x)=\alpha\ \text{and}\ \Tr(\gamma x)=0\},$$
so that $\text{wt}(c_{\gamma})=n_0-N_{\gamma}$. By the definition, we have
\begin{align*}
N_{\gamma}&=\frac{1}{p^{2}}\sum_{x\in\fq}\sum_{y\in\fp}\zeta_{p}^{y\big({\rm Tr}(x^{\frac{q-1}{2\ell^m}}+\beta x)-\alpha\big)}\sum_{z\in\fp}\zeta_{p}^{z{\rm Tr}(\gamma x)}\\
&=\frac{1}{p^{2}}\sum_{x\in\fq}\Big(\sum_{y\in\fp^*}\zeta_{p}^{y\big({\rm Tr}(x^{\frac{q-1}{2\ell^m}}+\beta x)-\alpha\big)}+1\Big)\Big(\sum_{z\in\fp^*}\zeta_{p}^{z{\rm Tr}(\gamma x)}+1\Big)\\
&=\frac{1}{p^{2}}\Big(q+\sum_{z\in\fp^*}\sum_{x\in\fq}\chi(z\gamma x)+
\sum_{y\in\fp^*}\zeta_{p}^{-y\alpha}\sum_{x\in\fq}\chi(yx^{\frac{q-1}{2\ell^m}}+y\beta x)\\
&\quad\quad+\sum_{y\in\fp^*}\zeta^{-y\alpha}\sum_{z\in\fp^*}\sum_{x\in\fq}\chi(yx^{\frac{q-1}{2\ell^m}}+(y\beta+z\gamma)x)\Big).
\end{align*}
Note that 
\begin{align*}
&\sum_{x\in\fq}\chi(z\gamma x)=0\ \text{for\ any}\ z\in\fp^*\ \text{and}\ \gamma\in\fq^{*},\\
&\sum_{y\in\fp^*}\zeta_{p}^{-y\alpha}\sum_{x\in\fq^{*}}\chi(yx^{\frac{q-1}{2\ell^m}}+y\beta x)=w(\alpha,\beta),\ \text{and}\\
&\sum_{y\in\fp^*}\zeta^{-y\alpha}\sum_{z\in\fp^*}\sum_{x\in\fq^{*}}\chi(yx^{\frac{q-1}{2\ell^m}}+(y\beta+z\gamma)x)=\sum_{z\in\fp^{*}}w(\alpha,\beta+z\gamma).
\end{align*}
It then implies that
$$N_{\gamma}=\frac{1}{p^{2}}\Big(q+p\sum_{y\in\fp^{*}}\zeta_{p}^{-y\alpha}+w(\alpha,\beta)+\sum_{z\in\fp^{*}}w(\alpha,\beta+z\gamma)\Big).$$
Therefore, the result follows immediately from $\text{wt}(c_{\gamma})=n_0-N_{\gamma}$. The proof of Proposition \ref{c3-3} is finished.
\end{proof}
From Propositions \ref{prop3.2} and \ref{prop3.3}, we obtain the main results of this paper as follows.
\begin{thm}\label{thm3.4}
Let $\alpha\in\fp, \beta\in\fq$, and let $\C_{{\alpha,\beta}}$ be the code whose defining set $D$ is given in \eqref{c1-2} with $N=2\ell^m$. Define $a_{1}=\ell^{m-1}(\ell-1)-\alpha$, $a_{2}=-\ell^{m-1}(\ell-1)-\alpha$, $a_{3}=\ell^{m-1}-\alpha$ and $a_{4}=-\ell^{m-1}-\alpha$. We have that $\C_{\alpha,\beta}$ is of dimension $(\ell-1)\ell^{m-1}$ except that $\alpha\ne0, \beta=0$ and none of $a_1, a_2, a_3, a_4$ are zero, in which case $\C_{\alpha,0}$ is an empty set. Precisely:
\begin{enumerate}[(a)]
\item Suppose $\alpha=\beta=0$. Then $\C_{0,0}$ is a two-weight code. Specifically:
\begin{enumerate}[(i)]
\item If $\ell\equiv 1\pmod{p}$, then $\C_{0,0}$ has length $\frac{(q-1)(\ell^m-\ell+1)}{\ell^m}$, with exactly two nonzero weights $w_1=\frac{p-1}{p\ell^m}((\ell^m-\ell+1)q-(\ell-1)\sqrt{q})$ and $w_2=\frac{(p-1)(\ell^m-\ell+1)}{p\ell^m}(q+\sqrt{q})$, where the weight enumerators are $A_{w_1} =\frac{(q-1)(\ell^m-\ell+1)}{\ell^m}$ and $A_{w_2}=\frac{(q-1)(\ell-1)}{\ell^m}$.
\item If $\ell\not\equiv 1\pmod{p}$ and $m\ge 2$, then $\C_{0,0}$ has length $\frac{(q-1)(\ell^{m-1}-1)}{\ell^{m-1}}$, with exactly two nonzero weights $w_1=\frac{(p-1)(\ell^{m-1}-1)}{p\ell^{m-1}}(q+\sqrt{q})$ and $w_2=\frac{p-1}{p\ell^{m-1}}((\ell^{m-1}-1)q-\sqrt{q})$, where the weight enumerators are $A_{w_1} =\frac{q-1}{\ell^{m-1}}$ and $A_{w_2}=\frac{(q-1)(\ell^{m-1}-1)}{\ell^{m-1}}$.
\end{enumerate}
\item Suppose $\alpha\ne 0$ and $\beta=0$. Then $\C_{\alpha,0}$ is a two-weight code if none of $a_1, a_2, a_3, a_4$ are zero. Furthermore,
\begin{enumerate}[(i)]
\item If either $a_1$ or $a_2$ is zero and $a_3\ne0, a_4\ne 0$, then then $\C_{\alpha,0}$ has length $\frac{q-1}{2\ell^{m}}$, with exactly two nonzero weights $w_1=\frac{p-1}{2p\ell^{m}}(q-(2\ell^m-1)\sqrt{q})$ and $w_2=\frac{p-1}{2p\ell^{m}}(q+\sqrt{q})$, where the weight enumerators are $A_{w_1} =\frac{q-1}{2\ell^{m}}$ and $A_{w_2}=\frac{(q-1)(2\ell^m-1)}{2\ell^{m}}$. 
\item If either $a_3$ or $a_4$ is zero and $a_1\ne0, a_2\ne 0$, then then $\C_{\alpha,0}$ has length $\frac{(q-1)(\ell-1)}{2\ell^{m}}$, with exactly two nonzero weights $w_1=\frac{(p-1)(\ell-1)}{2p\ell^{m}}(q+\sqrt{q})$ and $w_2=\frac{p-1}{2p\ell^{m}}((\ell-1)q-(2\ell^m-\ell+1)\sqrt{q})$, where the weight enumerators are $A_{w_1} =\frac{q-1}{2\ell^{m}}(2\ell^m-\ell+1)$ and $A_{w_2}=\frac{(q-1)(\ell-1)}{2\ell^{m}}$. 
\item If either $a_1=a_3=0, a_2\ne 0, a_4\ne 0$ or $a_2=a_4=0, a_1\ne 0, a_3\ne 0$, then then $\C_{\alpha,0}$ has length $\frac{q-1}{2\ell^{m-1}}$, with exactly two nonzero weights $w_1=\frac{p-1}{2p\ell^{m-1}}(q-(2\ell^{m-1}-1)\sqrt{q})$ and $w_2=\frac{p-1}{2p\ell^{m-1}}(q+\sqrt{q})$, where the weight enumerators are $A_{w_1} =\frac{q-1}{2\ell^{m-1}}$ and $A_{w_2}=\frac{q-1}{2\ell^{m-1}}(2\ell^{m-1}-1)$.  
\end{enumerate}
\item Suppose $\alpha=0$ and $\beta\ne0$. Then $\C_{0,\beta}$ has at most $p+1$ nonzero weights, denoted by $w_i$ where $1 \le i \le p+1$. Additionally, the length $n$ of $\C_{0,\beta}$ and all possible nonzero weights $w_i$ are given below.
\begin{enumerate}[(i)]
\item If $\ell\equiv 1\pmod{p}$ and $j_{\beta}\in J_1^{(2)}\cup J_3^{(2)}\cup J_1^{3)}\cup J_3^{(3)}$, then we have that $n=\frac{q}{p}-\frac{\sqrt{q}+1}{\ell^m}(\ell^n-\ell+1)$, $w_1=\frac{1}{p\ell^m}\left((\ell-1)q-(\ell^m-\ell+1)\sqrt{q}\right)$, and $w_{k+2}=\frac{p-1}{p^2}q-\frac{p-1-k}{p}\sqrt{q}$ with $0\le k\le p-1$.
\item If $\ell\equiv 1\pmod{p}$ and $j_{\beta}\in \{0\}\cup J_2\cup J_1^{(1)}\cup J_3^{(1)}$, then we have that $n=\frac{q-p}{p}+\frac{(\sqrt{q}+1)(\ell-1)}{\ell^m}$, $w_1=\frac{\ell-1}{p\ell^m}\left(q+\sqrt{q}\right)$, and $w_{k+2}=\frac{p-1}{p^2}q+\frac{p-1-k}{p}\sqrt{q}$ with $0\le k\le p-1$.
\item If \(\ell \not\equiv 1 \pmod{p}\) and \(j_{\beta} \in \{0, \ell^m\} \cup J_1^{(2)} \cup J_3^{(2)} \cup J_1^{(3)}\cup J_3^{(3)}\), 
then we have that \(n = \frac{q}{p} + \frac{(\sqrt{q} + 1)(-\ell^{m-1} + 1)}{\ell^{m-1}}\), 
\(w_1 = \frac{1}{p \ell^{m-1}} q + \frac{1 - \ell^{m-1}}{p \ell^{m-1}} \sqrt{q}\), 
and \(w_{k+2} = \frac{p-1}{p^2} q - \frac{p-1-k}{p} \sqrt{q}\) with \(0 \le k \le p-1\).
\item If \(\ell \not\equiv 1 \pmod{p}\) and \(j_{\beta} \in J_2\setminus\{\ell^m\} \cup J_1^{(1)} \cup J_3^{(1)}\), 
then we have that \(n = \frac{q}{p} + \frac{\sqrt{q} + 1}{\ell^{m-1}}-1\), 
\(w_1 = \frac{1}{p \ell^{m-1}} (q + \sqrt{q})\), 
and \(w_{k+2} = \frac{p-1}{p^2} q + \frac{p-1-k}{p} \sqrt{q}\) with \(0 \le k \le p-1\).
\end{enumerate}
\item Suppose $\alpha\ne 0$ and $\beta\ne 0$. If $a_1,a_2,a_3$ and $a_4$ are nonzero, then $\C_{\alpha,\beta}$ is a two-weight code of length $n=\frac{q}{p}$ having two nonzero weights $w_1=\frac{q}{p}$ and $w_2=\frac{(p-1)q}{p^2}$ with weight enumerators $A_{w_1}=p-1$ and $A_{w_2}=q-p$; otherwise, $\C_{\alpha,\beta}$ has at most $p+1$ nonzero weights, denoted by $w_i$, $1 \le i \le p+1$, and the length $n$ of $\C_{\alpha,\beta}$ and all possible nonzero weights $w_i$ are given below.
\begin{enumerate}[(i)]
\item Assume that either $a_1$ or $a_2$ is zero, and $a_3\ne 0, a_4\ne 0$. If $a_1=0$ and $j_{\beta}=0$, or $a_2=0$ and $j_{\beta}=\ell^m$, then \(n = \frac{q-p}{p} + \frac{(\sqrt{q} + 1)(2\ell^m-1)}{2\ell^{m}}\), 
\(w_1 = \frac{2\ell^m-1}{2p \ell^{m}} (q + \sqrt{q})\), 
and \(w_{k+2} = \frac{p-1}{p^2} q + \frac{p-1-k}{p} \sqrt{q}\) with \(0 \le k \le p-1\). If $a_1=0$ and $j_{\beta}\ne 0$, or $a_2=0$ and $j_{\beta}\ne \ell^m$, then \(n = \frac{q}{p}-\frac{\sqrt{q} + 1}{2\ell^{m}}\), 
\(w_1 = \frac{2\ell^m-1}{2p \ell^{m}}q-\frac{p}{2p\ell^m}\sqrt{q}\), 
and \(w_{k+2} = \frac{p-1}{p^2} q - \frac{p-1-k}{p} \sqrt{q}\) with \(0 \le k \le p-1\).
\item Assume that either $a_3$ or $a_4$ is zero, and $a_1\ne 0, a_2\ne 0$. If $a_3=0$ and $j_{\beta}\not\in J_1^{(2)}\cup J_3^{(2)}$, or $a_4=0$ and $j_{\beta}\not\in J_1^{(3)}\cup J_3^{(3)}$, then \(n = \frac{q}{p} - \frac{(\sqrt{q} + 1)(\ell-1)}{2\ell^{m}}\), 
\(w_1 = \frac{2\ell^m-\ell+1}{2p \ell^{m}} q - \frac{\ell-1}{2p\ell^m}\sqrt{q}\), 
and \(w_{k+2} = \frac{p-1}{p^2} q - \frac{p-1-k}{p} \sqrt{q}\) with \(0 \le k \le p-1\). If $a_3=0$ and $j_{\beta}\in J_1^{(2)}\cup J_3^{(2)}$, or $a_4=0$ and $j_{\beta}\in J_1^{(3)}\cup J_3^{(3)}$, then \(n = \frac{q}{p}+\frac{(\sqrt{q} + 1)(2\ell^m-\ell+1)}{2\ell^{m}}\), 
\(w_1 = \frac{2\ell^m-\ell+1}{2p \ell^{m}}(q+\sqrt{q})\), 
and \(w_{k+2} = \frac{p-1}{p^2} q + \frac{p-1-k}{p} \sqrt{q}\) with \(0 \le k \le p-1\).
\item Assume that either $a_1=a_3=0$ and $a_2\ne 0, a_4\ne 0$, or $a_2=a_4=0$ and $a_1\ne 0, a_3\ne 0$. If $a_1=a_3=0$ and $j_{\beta}\in\{0\}\cup J_1^{(2)}\cup J_3^{(2)}$, or $a_2=a_4=0$ and $j_{\beta}\in\{\ell^m\}\cup J_1^{(3)}\cup J_3^{(3)}$, then \(n = \frac{q-p}{p} + \frac{(\sqrt{q} + 1)(2\ell^{m-1}-1)}{2\ell^{m-1}}\), 
\(w_1 = \frac{2\ell^{m-1}-1}{2p \ell^{m-1}}(q +\sqrt{q})\), 
and \(w_{k+2} = \frac{p-1}{p^2} q + \frac{p-1-k}{p} \sqrt{q}\) with \(0 \le k \le p-1\). If $a_1=a_3=0$ and $j_{\beta}\not\in\{0\}\cup J_1^{(2)}\cup J_3^{(2)}$, or $a_2=a_4=0$ and $j_{\beta}\not\in\{\ell^m\}\cup J_1^{(3)}\cup J_3^{(3)}$, then \(n = \frac{q}{p}-\frac{\sqrt{q} + 1}{2\ell^{m-1}}\), 
\(w_1 = \frac{2\ell^{m-1}-1}{2p \ell^{m-1}}q-\frac{1}{2\ell^{m-1}}\sqrt{q}\), 
and \(w_{k+2} = \frac{p-1}{p^2} q - \frac{p-1-k}{p} \sqrt{q}\) with \(0 \le k \le p-1\).
\end{enumerate}
\end{enumerate}
\end{thm}
\begin{proof}
Let \(\alpha \in \mathbb{F}_p\) and \(\beta \in \mathbb{F}_q\). By combining the results of Propositions \ref{prop3.2} and \ref{prop3.3}, the length \(n\) of the code \(\mathcal{C}_{\alpha,\beta}\) can be directly determined. Specifically, \(n = 0\) if \(\alpha \neq 0\), \(\beta = 0\), and none of \(a_1, a_2, a_3, a_4\) are zero, while \(n > 0\) otherwise. In what follows, we assume \(D = \{d_1, \ldots, d_n\} \neq \emptyset\) and proceed to compute all nonzero weights. Define \(c_{\gamma} = (\text{Tr}(d_1\gamma), \ldots, \text{Tr}(d_n\gamma)) \in \mathcal{C}_{\alpha,\beta}\) for \(\gamma \in \mathbb{F}_q^*\).

Consider first the case where \(\beta = 0\). By Proposition \ref{prop3.3}(b), for any \(\gamma \in \mathbb{F}_q^*\), the weight of \(c_{\gamma}\) is given by
\begin{align}\label{c3-4}
\text{wt}(c_{\gamma}) = \frac{1}{p^2} \left( (p-1)q + (p-1)w(\alpha, 0) - (p-1)w(\alpha, \gamma) \right).
\end{align}
Substituting the possible values of \(w(\alpha, \gamma)\) into \eqref{c3-4}, we find that the set of weights \(\{\text{wt}(c_{\gamma}) : \gamma \in \mathbb{F}_q^*\} = \{w_1, w_2\}\), where \(w_1\) and \(w_2\) are as specified in Items (a) and (b). Thus, \(\mathcal{C}_{\alpha,\beta}\) is a two-weight code, and the frequencies of \(w_1\) and \(w_2\) can be determined by counting the number of \(\gamma\) such that \(\text{wt}(c_{\gamma}) = w_i\) for \(i = 1, 2\). This relies on the fact, demonstrated later, that the map \(\gamma \mapsto (\text{Tr}(d_1\gamma), \ldots, \text{Tr}(d_n\gamma))\) from \(\mathbb{F}_q\) to \(\mathcal{C}_{\alpha,\beta}\) is a bijection.

Next, assume \(\beta \neq 0\). From Proposition \ref{prop3.3}(b), for any \(\gamma \in \mathbb{F}_q^*\), the weight is expressed as
\begin{align}\label{c3-5}
\text{wt}(c_{\gamma}) = \frac{1}{p^2} \left( (p-1)q + (p-1)w(\alpha, \beta) - \sum_{z \in \mathbb{F}_p^*} w(\alpha, \beta + z\gamma) \right).
\end{align}
To determine the possible weights, we must evaluate the sum \(\sum_{z \in \mathbb{F}_p^*} w(\alpha, \beta + z\gamma)\) as \(\gamma\) varies over \(\mathbb{F}_q^*\). Define the set \(\Gamma = \{-z^{-1}\beta : z \in \mathbb{F}_p^*\}\). For \(\gamma \in \Gamma\), express \(\gamma = -z_0^{-1}\beta\) for some \(z_0 \in \mathbb{F}_p^*\). Then,
\[
\sum_{z \in \mathbb{F}_p^*} w(\alpha, \beta + z\gamma) = \sum_{z \in \mathbb{F}_p^*} w(\alpha, (1 - z z_0^{-1})\beta) = w(\alpha, 0) + \sum_{z \in \mathbb{F}_p^* \setminus \{z_0\}} w(\alpha, (1 - z z_0^{-1})\beta).
\]
Since \(w(\alpha, zx) = w(\alpha, x)\) for any \(z \in \mathbb{F}_p^*\) and \(x \in \mathbb{F}_q\), it follows that
\begin{align}\label{c3-6}
\sum_{z \in \mathbb{F}_p^*} w(\alpha, \beta + z\gamma) = w(\alpha, 0) + (p-2)w(\alpha, \beta)
\end{align}
for all \(\gamma \in \Gamma\). Combining this with Equation \eqref{c3-5} and Proposition \ref{prop3.2}, we can calculate the corresponding weight \(\text{wt}(c_{\gamma})\) for \(\gamma \in \Gamma\). Now, let \(\gamma \notin \Gamma\) and then \(\beta + z\gamma \neq 0\) for any \(z \in \mathbb{F}_p^*\). By Proposition \ref{prop3.2}, we know that \(w(\alpha, x)\) takes exactly two distinct values, denoted \(A\) and \(B\), as \(x\) ranges over \(\mathbb{F}_q^*\). So, if we define \(k_{\gamma}\) as the frequency of \(A\) in the sequence \(\{w(\alpha, \beta + z\gamma)\}_{z \in \mathbb{F}_p^*}\), then Equation \eqref{c3-5} gives that
\[
\wt(c_{\gamma}) = \frac{1}{p^2} \big( (p-1)q + (p-1-k_{\gamma})(A-B) \big)\  \text{if}\ w(\alpha, \beta) = A,
\]
and
\[
\wt(c_{\gamma}) = \frac{1}{p^2} \big( (p-1)q + (p-1-k_{\gamma})(B-A) \big)\  \text{if}\ w(\alpha, \beta) = B.
\]
Thus, we have determined all possible weights of \(c_{\gamma}\) as \(\gamma\) varies over \(\mathbb{F}_q^* \setminus \Gamma\).

Finally, we demonstrate that the map \(\sigma: \gamma \mapsto (\text{Tr}(d_1 \gamma), \ldots, \text{Tr}(d_n \gamma))\) from \(\mathbb{F}_q\) to \(\mathcal{C}_{\alpha, \beta}\) is a bijection. It suffices to prove that \(\sigma\) is injective. Suppose \(\sigma(x) = \sigma(y)\) for some \(x, y \in \mathbb{F}_q^*\). This implies:
\[
\sigma(x) - \sigma(y) = (\text{Tr}(d_1 (x-y)), \ldots, \text{Tr}(d_n (x-y))) =c_{x-y} = 0.
\]
From our earlier computations, we know that \(\wt(c_{\gamma}) \neq 0\) for all \(\gamma \in \mathbb{F}_q^*\). Hence, \(x - y = 0\), and thus \(x = y\), establishing injectivity. Therefore, $\#\mathcal{C}_{\alpha, \beta}=q$ and then $\mathcal{C}_{\alpha, \beta}$ is of dimension $(\ell-1)\ell^{m-1}$.

This completes the proof of Theorem \ref{thm3.4}.
\end{proof}

With the assistance of the mathematical software {\ttfamily SageMath}, we computed the following examples.

\begin{exa}\label{exa3.5}
Consider $p=3$, $\ell=7$, $m=1$, and $\alpha=\beta=1$. We verify that $\C_{\alpha,\beta}$ is a $[231,6]$ code with four distinct weights: $w_1=135$, $w_2=144$, $w_3=153$, and $w_4=162$, having weight enumerators $A_{w_1}=2$, $A_{w_2}=132$, $A_{w_3}=360$, and $A_{w_4}=234$.
\end{exa}

\begin{exa}\label{exa3.6}
Let $p=7$, $\ell=5$, $m=1$, and $\alpha=\beta=1$. We confirm that $\C_{\alpha,\beta}$ is a $[323,4]$ code with seven distinct weights: $w_1=203$, $w_2=259$, $w_3=266$, $w_4=273$, $w_5=280$, $w_6=287$, and $w_7=294$, with weight enumerators $A_{w_1}=6$, $A_{w_2}=96$, $A_{w_3}=420$, $A_{w_4}=432$, $A_{w_5}=924$, $A_{w_6}=456$, and $A_{w_7}=66$.
\end{exa}

\begin{exa}\label{exa3.7}
Take $p=7$, $\ell=5$, $m=1$, $\alpha=2$, and $\beta=1$. We establish that $\C_{\alpha,\beta}$ is a $[343,4]$ code with two distinct weights: $w_1=294$ and $w_2=343$, with weight enumerators $A_{w_1}=2394$ and $A_{w_2}=6$.
\end{exa}

\begin{exa}\label{exa3.8}
Set $p=7$, $\ell=5$, $m=1$, $\alpha=1$, and $\beta=0$. We determine that $\C_{\alpha,\beta}$ is a $[960,4]$ code with two distinct weights: $w_1=798$ and $w_2=840$, with weight enumerators $A_{w_1}=960$ and $A_{w_2}=1440$.
\end{exa}

Theorem \ref{thm3.4} asserts that $\C_{\alpha,\beta}$ has at most $p+1$ nonzero weights. Example \ref{exa3.5} demonstrates that $p+1$ serves as a tight upper bound for the number of nonzero weights.

\section{The dual codes}
For $\alpha \in \fp$ and $\beta \in \fq$, let $\C^{\perp}_{\alpha,\beta}$ denote the dual code of $\C_{\alpha,\beta}$, defined as
$$\C^{\perp}_{\alpha,\beta}=\{c\in\fp^n: \<c,x\>=0\ \text{for\ all}\ x\in\C_{\alpha,\beta}\},$$
where $n$ represents the length of $\C_{\alpha,\beta}$, and $\langle c, x \rangle$ is the Euclidean inner product of $c$ and $x$. 
\begin{thm}\label{thm4.1}
Let $\beta = 0$ and $n$ be the length of $\C_{\alpha,0}$. If $n > 0$, then $\C_{\alpha,0}^{\perp}$ has dimension $n - (\ell - 1)\ell^{m-1}$ and minimum distance $2$.
\end{thm}

\begin{proof}
The dimension follows directly from Theorem \ref{thm3.4}. We now determine the minimum distance $d$ of $\C_{\alpha,0}^{\perp}$. First, we note that $d \geq 2$, which is guaranteed by the surjectivity of the trace function. Rewrite the set $D = \{x \in \mathbb{F}_q^* : \textup{Tr}(x^{\frac{q-1}{2\ell^m}}) = \alpha\}$ as $D = \{d_1, d_2, \ldots, d_n\}$. For any $d_i \in D$, it follows that $-d_i \in D$ as well, since the trace function satisfies $\textup{Tr}((-d_i)^{\frac{q-1}{2\ell^m}}) = \textup{Tr}(d_i^{\frac{q-1}{2\ell^m}}) = \alpha$. Moreover, $d_i \neq -d_i$, implying that $-d_i = d_j$ for some $j \neq i$. Consider the codeword $c = (c_1, c_2, \ldots, c_n)$, where $c_i = c_j = 1$ and $c_k = 0$ for all $k \neq i, j$. This codeword belongs to $\C_{\alpha,0}^{\perp}$ and has weight two. Hence, $d \leq 2$, and combined with $d \geq 2$, we conclude that $d = 2$.
\end{proof}
\begin{rem}
Let $d$ be the minimum distance of $\C_{\alpha,\beta}^{\perp}$. For $\beta \neq 0$, we have $d \geq 2$. Moreover, we can establish a necessary and sufficient condition for $d = 2$. Suppose $c = (c_1, \dots, c_n) \in \C_{\alpha,\beta}^{\perp}$ has $c_i \neq 0$ and $c_j \neq 0$ for some $i \neq j$, with $c_k = 0$ for all $k \neq i, j$. Then, for all $x \in \fq$, we have
\[
c_i{\rm Tr}(d_i x) + c_j{\rm Tr}(d_j x) = 0.
\]
This implies $c_i d_i + c_j d_j = 0$. Since $(p-1) \mid \frac{q-1}{2\ell^m}$, it follows that $d_i$ and $d_j$ belong to the set
\[
Y = \{ x \in \fq^* : \Tr(x^{\frac{q-1}{2\ell^m}}) = \alpha, \, \Tr(\beta x) = 0 \}.
\]
Thus, $\# Y > 0$. Conversely, if $\# Y > 0$, then $d = 2$. Hence, $d = 2$ if and only if $\# Y > 0$. Furthermore,
\[
\# Y = \frac{1}{p^2} \sum_{x \in \fq^*} \sum_{y \in \fp} \zeta_p^{y {\rm Tr}(x^{\frac{q-1}{2\ell^m}}) - y \alpha} \sum_{z \in \fp} \zeta_p^{z {\rm Tr}(\beta x)}.
\]
This simplifies to $\# Y = \frac{1}{p^2} \big( w(\alpha, 0) + (p-1) w(\alpha, \beta) + q - p \big)$. Therefore, by Proposition \ref{prop3.2}, we can obtain a necessary and sufficient condition for $d=2$.
\end{rem}
\begin{defn}\label{defn4.3}
A code $\mathcal{C}$ with parameters $[n,k,d]$ is called \textit{optimal} if there does not exist an $[n,k,d+1]$ code.
\end{defn}
Using the sphere packing bound and straightforward verification, the final result of this paper follows directly.

\begin{cor}\label{cor4.4}
Let $\alpha = \beta = 0$, and let $n$ be the length of the code $\mathcal{C}_{\alpha,\beta}^{\perp}$. Then $\mathcal{C}_{\alpha,\beta}^{\perp}$ is optimal for $n > 0$.
\end{cor}

\begin{proof}
Denote by $d$ the minimum distance of the dual code $\mathcal{C}_{\alpha,\beta}^{\perp}$. By the sphere packing bound,
\[
|\mathcal{C}_{\alpha,\beta}^{\perp}| \leq \frac{p^n}{\sum_{i=0}^{t} \binom{n}{i} (p-1)^i},
\]
where $t = \lfloor \frac{d-1}{2} \rfloor$ is the greatest integer less than or equal to $\frac{d-1}{2}$. Since $d = 2$ by Theorem \ref{thm4.1}, the dual code $\mathcal{C}_{\alpha,\beta}^{\perp}$ is optimal whenever $n > \frac{q-1}{p-1}$. Thus, the result follows immediately from Theorem \ref{thm3.4}.
\end{proof}

\section{Conclusion}
In this paper, we have conducted a comprehensive analysis of the weight distribution of the $p$-ary linear code $\mathcal{C}_{\alpha, \beta}$, defined over the finite field $\fq$ with $q = p^{\phi(N)}$, where $p$ is an odd prime and $N = 2\ell^m$ such that $p$ is a primitive root modulo $N$. By leveraging exponential sums, notably the binomial Weil sums evaluated through the results of \cite{[CG]}, we studied the weight distribution of $\C_{\alpha, \beta}$ for any $\alpha \in \mathbb{F}_p$ and $\beta \in \mathbb{F}_q$. Our findings generalize previous results for specific primes (e.g., $p = 2$ in \cite{[WDX]} and $p = 3$ in \cite{[Che]}) to the broader context of any odd prime $p$, offering a unified framework for understanding this class of codes. Moreover, we explored the dual code $\mathcal{C}_{\alpha, \beta}^{\perp}$, demonstrating that $\mathcal{C}_{0, 0}^{\perp}$ has minimum distance 2 and is optimal with respect to the sphere packing bound.

One significant application of a linear code $\C$ over $\fp$ is in constructing secret sharing schemes, as exemplified in \cite{[Sha]}. Let $\wt_{\min}$ and $\wt_{\max}$ denote the minimum and maximum nonzero weights of the $p$-ary linear code $\C$. If $\frac{\wt_{\min}}{\wt_{\max}} > \frac{p-1}{p}$, then $\mathcal{C}$ can be utilized to design secret sharing schemes with notable access structures \cite{[YD]}. From Theorem \ref{thm3.4}, we verify that many of the codes presented in this paper satisfy the condition
$$\frac{\wt_{\min}}{\wt_{\max}}>\frac{p-1}{p}$$
under certain natural conditions. For instance, in Example \ref{exa3.8}, we find $\frac{\wt_{\min}}{\wt_{\max}} = \frac{798}{840} >\frac{6}{7}$.

\section*{acknowledgment}
The first author would like to sincerely thank Professor Arne Winterhof for his invaluable guidance and support during the first author's visit at the RICAM Institute of the Austrian Academy of Sciences.


\end{document}